\documentclass[11pt]{article}
\usepackage{amsmath}
\usepackage{amsfonts}
\usepackage{amsthm}
\usepackage{amssymb, setspace}
\usepackage{color,graphicx,epsfig,geometry,hyperref,fancyhdr}
\usepackage[T1]{fontenc}
\usepackage{float}
\usepackage{subfig}
\usepackage{comment}
\usepackage{commath}
\usepackage{bbm}
\usepackage{listings}
\usepackage{mathrsfs,fleqn}
\newtheorem{theorem}{Theorem}[section]

\newtheorem{corollary}{Corollary}[section]

\newtheorem{remark}{Remark}[section]

\newcommand{\R}{\mathbb{R}}
\newcommand{\C}{\mathbb{C}}

\graphicspath{{paper-pics/}}

\begin{document}

\begin{flushleft}
\Large 
\noindent{\bf \Large Inverse parameter and shape problem for an isotropic scatterer with two conductivity coefficients}
\end{flushleft}

\vspace{0.2in}

{\bf  \large Rafael Ceja Ayala and Isaac Harris\footnote{Orcid:0000-0002-9648-6476 and corresponding author}}\\
\indent {\small Department of Mathematics, Purdue University, West Lafayette, IN 47907 }\\
\indent {\small Email:  \texttt{rcejaaya@purdue.edu}  and \texttt{harri814@purdue.edu} }\\

{\bf  \large Andreas Kleefeld\footnote{Orcid:0000-0001-8324-821X}}\\
\indent {\small Forschungszentrum J\"{u}lich GmbH, J\"{u}lich Supercomputing Centre, } \\
\indent {\small Wilhelm-Johnen-Stra{\ss}e, 52425 J\"{u}lich, Germany}\\
\indent {\small University of Applied Sciences Aachen, Faculty of Medical Engineering and } \\
\indent {\small Technomathematics, Heinrich-Mu\ss{}mann-Str. 1, 52428 J\"{u}lich, Germany}\\
\indent {\small Email: \texttt{a.kleefeld@fz-juelich.de}}\\


\begin{abstract}
\noindent In this paper, we consider the direct and inverse problem for isotropic scatterers with two conductive boundary conditions. First, we show the uniqueness for recovering the coefficients from the known far-field data at a fixed incident direction for multiple frequencies. Then, we address the inverse shape problem for recovering the scatterer for the measured far-field data at a fixed frequency. Furthermore, we examine the direct sampling method for recovering the scatterer by studying the factorization for the far-field operator. The direct sampling method is stable with respect to noisy data and valid in two dimensions for partial aperture data. The theoretical results are verified with numerical examples to analyze the performance by the direct sampling method. 
\end{abstract}

\section{Introduction}
In this paper, we investigate the inverse isotropic scattering problem, which involves determining an inhomogeneous medium from measurements obtained far from the scatterer. The problem focuses on two conductivity parameters, denoted as $\eta$ and $\lambda$, previously explored in \cite{te-2cbc} from the perspective of a transmission eigenvalue problem. These physical parameters model an object with a thin layer covering the exterior, such as an aluminum sheet. Assuming a fixed wave number, we work with the measured far-field pattern and employ qualitative reconstruction methods to recover the unknown scatterer(s). Qualitative reconstruction methods, widely used in various inverse scattering problems \cite{ammari1,GLSM,fmconductbc,harris-nguyen,w/Jake,RTM-gibc,Liu,nf-meng,phong1,GLSM2}, require minimal a priori information about the scatterer, making them advantageous for nondestructive testing in fields such as medical imaging and engineering.

Our approach extends the application of the Direct Sampling Method (DSM) in order to reconstruct scatterers with two conductivity parameters. Here, we assume to have access to the far-field operator, i.e., the measured far-field pattern for all sources and receivers along the unit circle/sphere. The paper emphasizes the application of the DSM for solving inverse shape problems in scattering theory, as evidenced in related works \cite{BS1,BS2,nf-harris,Ito2012,Ito2013,Kang1,Kang2,Kang3,liem}. Additionally, we address a limited aperture problem using the DSM, considering a partial number of incident and/or observation directions in $\mathbb{R}^2$. Despite having limited data, we demonstrate that accurate scatterer recovery is achievable with partial information along the incident and/or observation directions, which is a common problem in engineering and medical imaging scenarios where data may be restricted.

The final component of our study focuses on the unique recovery of physical parameters $n$ and $\eta$ for the inverse problem and where a fixed boundary parameter $\lambda$ is considered. While uniqueness studies are typically challenging, we establish a connection from the inverse problem to a transmission eigenvalue problem with respect to two conductivity parameters. With the help of the transmission eigenvalue problem, one is able to establish discreteness to aid the uniqueness recovery of the physical parameters $n$ and $\eta$ in the presence of fixed boundary parameter $\lambda$. This discreteness argument plays a crucial role in proving the uniqueness result for the physical parameters $\eta$ and $n$.

The paper is organized as follows: the next section outlines the direct and inverse problems considered, detailing the scattering by an isotropic scatterer with two conductive boundary conditions. Section 3 delves into the uniqueness problem of a variable refractive index $n$ and complex constant conductivity parameter $\eta$, exploring the discreteness of the transmission eigenvalue problem. In Section 4, the factorization of the far-field operator is discussed, enabling the reconstruction of the medium with two boundary parameters for both full and limited apertures. Finally, numerical examples based on the DSM are presented in the last section.

\section{The direct scattering problem}\label{direct-prob-LDSM}
In this section, we formulate the direct scattering problem in $\mathbb{R}^d$ for isotropic scatterers with two conductive boundary conditions where $d=2,3$. Here, we assume that an incident plane wave $u^i(x, \hat{y})=\text{e}^{\text{i}kx\cdot\hat{y}}$ is used to illuminate the scatterer for an incident direction $\hat{y}\in \mathbb{S}^{d-1}$(i.e. unit circle/sphere) and point $x\in \mathbb{R}^d$. We will assume the scatterer $D \subset \mathbb{R}^d$ has a boundary that is $C^2$ where $\nu$ denotes the unit outward normal vector to $\partial D$.  Notice, that the incident field $u^i$ satisfies the Helmholtz equation in all of $\R^d$. The interaction of the incident field and the scatterer produces the radiating scattered field $u^s(x,\hat{y})\in H^1_{loc}(\mathbb{R}^d)$ that satisfies 
\begin{align}
\Delta u^s+k^2 n(x)u^s=k^2 \big(1-n(x)\big) u^i \quad & \text{in} \hspace{.2cm}  \mathbb{R}^d \backslash \partial D \label{direct1}\\
u^s_{-} - u_{+}^s=0  \quad \text{ and } \quad \lambda \partial_\nu \big(u^s_{-} +u^i\big) = \eta(x) \big(u_{+}^s+u^i\big) +  \partial_\nu \big(u_{+}^s+u^i\big)  \quad & \text{on} \hspace{.2cm} \partial D \label{direct2}
\end{align}
along with the Sommerfeld radiation condition 
\begin{align}
{\partial_r u^s} - \text{i} ku^s =\mathcal{O} \left( \frac{1}{ r^{(d+1)/2} }\right) \quad \text{ as } \quad r \rightarrow \infty. \label{direct3}
\end{align}
We let  $\partial_\nu \phi:= \nu \cdot \nabla \phi$ for any $\phi$ and $r:=|x|$. The radiation condition \eqref{direct3} is satisfied uniformly in all directions $\hat{x}:=x/|x|$ where $k>0$ denotes the wave number.  Here $-$ and $+$ corresponds to taking the trace from the interior or exterior of $D$, respectively. We also note that $u^i$ and its normal derivative are continuous across the boundary of $\partial D$. 

We assume that the refractive index $n \in L^{\infty}(D)$ satisfies that supp$(n-1)=D$ and the conductivity $\eta \in L^{\infty}(\partial D)$ where
$$Im(n(x))\geq0 \quad \text{ for a.e. $x \in  D$ } \quad \text{and} \quad Im(\eta(x))\geq 0 \quad \text{ for a.e.  $x \in \partial D.$}$$
The second boundary parameter $\lambda$ is a fixed complex-valued constant. Note, that the well-posedness of \eqref{direct1}--\eqref{direct3} was established in \cite{fmconductbc}.

Due to the fact that $u^s$ is a radiating solution to the Helmholtz equation on the exterior of $D$, similarly we have that (see for e.g. \cite{approach,IST})
$$u^s(x,\hat{y})=\gamma \frac{\text{e}^{\text{i}k|x|}}{|x|^{(d-1)/2}}\left\{u^{\infty}(\hat{x},\hat{y})+\mathcal{O}\left(\frac{1}{|x|}\right) \right\} \quad \text{as} \quad |x| \longrightarrow \infty$$
uniformly with respect to $\hat{x}$ where   
$$ \gamma=\frac{\text{e}^{\text{i}\pi/4}}{\sqrt{8\pi k}}\hspace{.3cm}\text{in}\hspace{.3cm}\mathbb{R}^2\hspace{.3cm}\text{and}\hspace{.3cm}\gamma=\frac{1}{4\pi}\hspace{.3cm}\text{in}\hspace{.3cm}\mathbb{R}^3.$$
The function $u^{\infty}(\hat{x},\hat{y})$ denotes the far-field pattern of the scattered field for observation direction $\hat{x}$ and incident direction $\hat{y}$ on $\mathbb{S}^{d-1}$ as well as the wave number $k>0$. We can now define the far-field operator denoted $F$ given by 
\begin{equation}\label{ff-operator}
(Fg)(\hat{x})=\int_{\mathbb{S}^{d-1}}u^{\infty}(\hat{x},\hat{y})g(\hat{y})\, \text{d}s(\hat{y}) \quad \text{ for } \quad g\in L^2(\mathbb{S}^{d-1})
\end{equation}
mapping $L^2(\mathbb{S}^{d-1})$ into itself.

We are interested in using the known far-field pattern to study the inverse shape and parameter problems. To this end, we first need to show that $u^s$ along with \eqref{direct1}--\eqref{direct3} satisfies a Lippmann-Schwinger integral equation. The radiating fundamental solution for the Helmholtz equation is given by 
\begin{equation}\label{funsolu}
\Phi_{k}(x,y)=
    \begin{cases}
        \frac{\text{i}}{4} H_0^{(1)}(k|x-y|) & d=2\\
        \frac{\text{e}^{\text{i}k|x-y|}}{4\pi|x-y|} & d=3
    \end{cases}
\end{equation}
where $H^{(1)}_0$ denotes the first kind Hankel function of order zero. Using Green's 2nd Theorem when $x$ is in the interior of $D$ gives that 
\begin{align*}
   u^s(x)\chi_D&=-\int_D \Phi_{k}(x,z) [\Delta u^s(z)+k^2 u^s(z)] \, \text{d}z\\
&\hspace{0.5in}+\int_{\partial D} \Phi_{k}(x,z)\partial_\nu u^s_-(z)-u^s_-(z)\partial_\nu \Phi_{k}(x,z)\, \text{d}s(z)
\end{align*}
where $\chi_D$ is the indicator function on the scatterer $D$. In a similar manner, using Green's 2nd Theorem when $x$ is in the exterior of $D$  gives that
\begin{align*}
    u^s(x)(1-\chi_D)&=-\int_{\partial D} \Phi_{k}(x,z)\partial_\nu u^s_+(z)-u^s_+(z)\partial_\nu \Phi_{k}(x,z) \, \text{d}s(z)\\
    &\hspace{0.5in}+\int_{\partial B_R} \Phi_{k}(x,z)\partial_\nu u^s(z)-u^s(z)\partial_\nu \Phi_{k}(x,z)\, \text{d}s(z).
\end{align*}
where $B_R=\{ x \in \R^d \, \, : \, \, |x|<R \}$ such that $D \subset B_R$. 
Therefore, by adding the above expressions we obtain the Lippmann-Schwinger type representation of the scattered field 
\begin{align}
u^s(x)&= k^2 \int_{D} (n-1)\Phi_{k}(x,z)\big[u^s(z)+u^i(z)\big] \, \text{d}z + \int_{\partial D} \eta\Phi_{k}(x,z)\big[u^s(z)+u^i(z)\big] \, \text{d}s(z) \nonumber \\ 
& \hspace{1.5in}+\int_{\partial D}(1-\lambda)\Phi_{k}(x,z)\big[\partial_\nu u^s_{-} (z) +\partial_\nu u^i(z)  \big] \,\text{d}s(z). \label{totalfield}
\end{align}
Notice, that we have used the scattered field and fundamental solution satisfying the radiation condition \eqref{direct3} to handle the boundary integral over $\partial B_R$ by letting $R \to \infty$. Equation \eqref{totalfield} will be vital for the analysis in later sections.

\begin{section}{Uniqueness result for the coefficients }
In this section, we will prove a uniqueness result for recovering constant $n$ and $\eta$ from the far-field data for one incident direction with multi-frequency data. Recently in \cite{Uniqueone,Roland1,Roland2,uniquenessnobc,multifre,uniqueoneboun,uniqueoneboun2} the authors have considered similar inverse scattering problems with multi-frequency far-field data. Here we will assume that the scatterer $D$ and parameter $\lambda$ (also constant) are known. Therefore, we assume that we have the penetrable scatterer $(D,n_j,\eta_j,\lambda)$ for $j=1,2$ that produces the far-field data $u^{\infty}_j(\hat{x}; k)$ where the incident direction $\hat{y}$ is fixed. We will also assume that we have the far-field data for a range of wave numbers i.e. for all $k \in (k_{min} , k_{max}) \subset \R_{+}$ where $k_{min} < k_{max}$. 

In order to prove our uniqueness result, we will need to consider the corresponding transmission eigenvalue problem that arises when we assume that 
$$u^{\infty}_1(\hat{x}; k) = u^{\infty}_2(\hat{x}; k) \quad \text{for all $\hat{x} \in \mathbb{S}^{d-1}$ and $k\in (k_{min} , k_{max})$. } $$ 
Since the far-field data for the scatterers $(D,n_j,\eta_j,\lambda)$ are equal we have that by Rellich's lemma the corresponding total fields $u^{s}_1 = u^{s}_2$ for all $x \in \R^d \setminus \overline{D}$. Therefore, by \eqref{direct1}--\eqref{direct3} we have that in the scatterer $D$ the total fields satisfy 
\begin{align*}
\Delta u^{(1)}+k^2 n_1 u^{(1)} = 0 \quad \text{ and } \quad  \Delta u^{(2)}+k^2n_2 u^{(2)} = 0  \hspace{.2cm}& \text{in} \hspace{.2cm}  D \\
u^{(1)}=u^{(2)} \quad \text{ and } \quad   \lambda \left( \partial_{\nu}u^{(1)}-\partial_{\nu}u^{(2)} \right) =(\eta_1-\eta_2)u^{(2)}  \hspace{.2cm}& \text{on} \hspace{.2cm}  \partial D.
\end{align*}
Here $u^{(j)}=u_j^s+u^i$ corresponds to the total field for the scattering problem associated with the scatterers $(D,n_j,\eta_j,\lambda)$. This homogeneous system is similar to the transmission eigenvalue problem studied in \cite{te-cbc}. The main difference is the fact that there are two sets of parameters $(n_j,\eta_j)$ for $j=1,2$. We now turn our attention to proving that the set of wave numbers $k \in \C$ such that the above system has a non-trivial solution is at most discrete. 

\subsection{Discreteness  of Transmission Eigenvalue Problem} 
Now, we consider the above system under the weaker assumptions that 
$$|n_1-n_2|^{-1}\in L^{\infty}(D) \quad \text{and} \quad |\eta_1-\eta_2|^{-1}\in L^{\infty}(\partial D).$$
Notice that this assumption allows for the contrasts $n_1-n_2$ and $\eta_1-\eta_2$ to change sign in $D$ and $\partial D$, receptively. The new transmission eigenvalue problem under consideration is given by: find $k\in \mathbb{C}$ and nontrivial $(w,v)\in L^2(D)\times L^2(D)$ such that  
\begin{align}
\Delta w+k^2 n_1 w = 0 \quad \text{ and } \quad  \Delta v+k^2n_2v = 0  \hspace{.2cm}& \text{in} \hspace{.2cm}  D \label{eigproblemvariable1}\\
w=v \quad \text{ and } \quad   \lambda \partial_{\nu}w= \lambda \partial_{\nu}v+(\eta_1-\eta_2) v \label{eigproblemvariable2} \hspace{.2cm}& \text{on} \hspace{.2cm}  \partial D
\end{align}
where $w-v\in X(D)$. Here, the variational space for the difference of the eigenfunctions is defined as 
$$X(D)=H^1_0(D)\cap H^2(D)\hspace{.5cm}\text{such that}\hspace{.5cm} \| \cdot \|_{X(D)}=\|\Delta \cdot \|_{L^2(D)}.$$
By our assumption that the boundary $\partial D$ is $C^2$ gives that the $L^2(D)$ norm of the Laplacian is equivalent to the $H^2(D)$ norm in the associated Hilbert space $X(D).$

Just as in \cite{te-cbc,te-cbc2,te-cbc3} we will consider a 4th order formulation of \eqref{eigproblemvariable1}--\eqref{eigproblemvariable2}. To this end, we let $u=w-v$ which implies that 
\begin{align}
\Delta u+k^2 n_1 u = -k^2(n_1-n_2)v \quad  \hspace{.2cm}& \text{in} \hspace{.2cm}  D \label{eigproblemu1}\\
u=0 \quad \text{ and } \quad   \lambda \partial_{\nu}u=(\eta_1-\eta_2) v \label{eigproblemu2} \hspace{.2cm}& \text{on} \hspace{.2cm}  \partial D.
\end{align}
Now, we can formulate an analogous transmission eigenvalue problem with a conductive boundary condition which is written as: find $k\in \mathbb{C}$ values such that there is a nontrivial solution $u\in X(D)$ satisfying 
\begin{align}
(\Delta +k^2 n_2) \frac{1}{n_1-n_2}(\Delta u+k^2 n_1 u) = 0 \quad  \hspace{.2cm}& \text{in} \hspace{.2cm}  D \label{traceeigproblem1}\\
 \lambda \partial_{\nu}u=(\eta_1-\eta_2) v \quad & \text{on} \hspace{.2cm}  \partial D,\label{traceeigproblem2}
\end{align}
where equation \eqref{traceeigproblem2} is understood in the sense of the Trace Theorem and we have that 
$$v=-\frac{1}{k^2(n_1-n_2)}(\Delta u+k^2n_1u)\hspace{.3cm}\text{and}\hspace{.3cm}w=-\frac{1}{k^2(n_1-n_2)}(\Delta u+k^2n_2u).$$
It is clear that \eqref{eigproblemvariable1}--\eqref{eigproblemvariable2} and  \eqref{traceeigproblem1}--\eqref{traceeigproblem2} are equivalent. 
To proceed, we must come up with its respective variational formulation and exploit it. Taking a function $\phi \in X(D)$, multiply it by the conjugate of \eqref{traceeigproblem1}, and use Green's 2nd Theorem to get 
\begin{equation}
 \int_{\partial D}\frac{\lambda k^2}{\eta_1-\eta_2}\partial_\nu u \partial_\nu \overline{\phi}\,\text{d}s(z)+\int_{D}\frac{1}{n_1-n_2}(\Delta u+k^2n_1u)(\Delta \overline{\phi}+k^2n_2\overline{\phi})\,\text{d}z=0 \label{variationaltwovari}
\end{equation}
We will use \eqref{variationaltwovari} to prove discreetness for the set of eigenvalues $k \in \C$ for which there is a non-trivial solution to \eqref{eigproblemvariable1}--\eqref{eigproblemvariable2}. 

To this end, we will use a $T$--coercivity argument studied in \cite{sign-chan} as well as the Analytic Fredholm theorem (see for e.g. \cite{IST}) to prove the discreteness of the transmission eigenvalues. We say that the sesquilinear form $a(\cdot, \cdot)$ is $T$--coercive on $X(D)$ if there is an isomorphism $T:X(D)\longmapsto X(D)$ such that $a(\cdot, T\cdot)$ is a coercive sesquilinear form. Using the inf-sup condition (see for e.g. \cite{inf-sup} for more details) in this case, we have that if $a(\cdot,\cdot)$ is $T$--coercive, then $a(\cdot,\cdot)$ can be represented by a continuous invertible operator $A: X(D)\longmapsto X(D)$ where 
$$a(u,\phi)=(Au,\phi)_{X(D)}\hspace{.3cm}\text{for all }\hspace{.3cm}u,\phi\in X(D).$$
We are going to use this definition to split the variational formulation \eqref{variationaltwovari} into a compact and an invertible part. Therefore, one has that the variational formulation \eqref{variationaltwovari} is given by 
\begin{equation}
    a(u,\phi)+b_k(u,\phi)=0\hspace{.3cm}\text{for all}\hspace{.3cm} \phi\in X(D). \label{variationalab}
\end{equation}
We can show that $a(\cdot, \cdot)$ and $b_k(\cdot, \cdot)$ are both bounded on $X(D)\times X(D)$ and have the form 
\begin{equation}
    a(u,\phi)=\int_D \frac{1}{n_1-n_2}\Delta u\Delta \overline{\phi}\,\text{d}z
\end{equation}
and 
\begin{align}
b_k(u,\phi)&=k^2\int_D\frac{1}{n_1-n_2}(\overline{\phi}\Delta u+u\Delta\overline{\phi})\,\text{d}z-k^2\int_D\nabla u\cdot\nabla \overline{\phi}\,\text{d}z \nonumber \\
    &\hspace{2.5cm}+ \lambda k^2\int_{\partial D}\frac{1}{\eta_1-\eta_2}\partial_\nu u \partial_\nu \overline{\phi}\,\text{d}s(z)+k^4\int_{D}u\overline{\phi}\,\text{d}z,
\end{align}
where $b_k(u,\phi)$ is the result of using Green's 1st Theorem and adding and subtracting like terms. The boundedness of both sesquilinear forms comes from the assumptions on the physical parameters and the Trace Theorem.

Just as in \cite{te-cbc} we can conclude that $b_k(\cdot,\cdot)$ can be represented by a compact operator $ B_k: X(D)\longmapsto X(D)$. Thus we obtain that 
$$b_k(u,\phi)=(B_ku,\phi)_{X(D)}\hspace{.3cm}\text{for all }\hspace{.3cm}u,\phi\in X(D)$$
where $B_k$ clearly depends analytically on $k\in \mathbb{C}$. We can also establish the existence of a bounded linear operator $A:X(D)\longmapsto X(D)$ that represents the sesquilinear form $a(\cdot,\cdot).$ Following the analysis in \cite{sign-chan}, we define $T:X(D)\longmapsto X(D)$ such that 
\begin{equation}
    \Delta T\phi=(n_1-n_2)\Delta\phi \hspace{.3cm}\text{for all }\hspace{.3cm}\phi\in X(D) \label{Toperator}.
\end{equation}
Now that we have a hold of our operator $T,$ we aim to show that this operator is an isomorphism on the space $X(D)$. We first notice that by the well-posedness of the Poisson problem we have that for every $\phi\in X(D),$ there is the existence of $T\phi\in H^1_0(D)$ satisfying \eqref{Toperator}. Using elliptic regularity as in \cite{evansbook}, we can then conclude that $T\phi \in X(D)$. The last item we need in order to proceed with our goal is a bound under the conditions on the parameters and the definition of $T$ in \eqref{Toperator}. To this end, by \eqref{Toperator} and $|n_1-n_2|^{-1}\in L^{\infty}(D)$ we have 
$$C\|\Delta\phi\|^2_{L^2(D)}\leq \big|(\Delta T\phi, \Delta \phi)_{L^2(D)}\big|$$
where $C$ does not depend on $\phi$ but only on the contrast $n_1-n_2$. This gives that the operator $T$ is indeed coercive on the space $X(D)$ and as a consequence we have an isomorphism by appealing to the Lax-Milgram Lemma. The existence of an isomorphism in terms of the operator $T$ is the set up for our main result in terms of the uniqueness of the refractive index. Thus, we are ready to prove our discreteness result.

\begin{theorem}
    Assume that $|n_1-n_2|^{-1}\in L^{\infty}(D)$ and $|\eta_1-\eta_2|^{-1}\in L^{\infty}(\partial D)$, then the set of transmission eigenvalues in \eqref{traceeigproblem1}--\eqref{traceeigproblem2} is at most discrete.
\end{theorem}
\begin{proof}
    As mentioned before, we will appeal to the Analytic Fredholm Theorem to prove this result. We will use the definition of $T$ in \eqref{Toperator} and make the following observation 
    $$a(u,Tu)=\int_D\frac{1}{n_1-n_2}\Delta u\Delta\overline{Tu}\,\text{d}z=\|\Delta u\|^2_{L^2(D)}\hspace{.3cm}\text{for all}\hspace{.3cm}u\in X(D).$$
    As a consequence of this observation, we have that $a(\cdot,\cdot)$ is $T$--coercive on $X(D)$ and thus $a(\cdot,\cdot)$ can be represented by an invertible operator. By \eqref{variationalab}, we have that $k\in\mathbb{C}$ is a transmission eigenvalue if and only if $A+B_k$ is not injective. We know that $A$ is invertible and $B_k$ is compact by definition and this implies that $A+B_k$ is Fredholm with index zero. In our case, we denote $B_0$ to be the zero operator where $k=0.$ Then at $k=0$, we have that $A+B_0 = A$ is injective and by the Analytic Fredholm Theorem there is at most a discrete set of values in $\mathbb{C}$ where $A+B_k$ fails to be injective. This proves the claim. 
\end{proof}

\subsection{Uniqueness for two of the material parameters}
In this section, we present a uniqueness result in determining a constant refractive index as well as conductivity parameter by using the fixed incident direction far-field pattern. In the previous section, we proved that the set of transmission eigenvalues is discrete with respect to our problem. We will be able to relate such result to the study on uniqueness of the inverse scattering problem. Moreover, we will be able to relate it to the interior transmission eigenvalue problem dealing with two conductivity parameters. This means that we will consider the transmission eigenvalue problem with two conductivity parameters and establish a relationship to the uniqueness of the variable refractive index and complex constant conductivity parameter.
\begin{theorem}
    Let $u^{\infty}_j(\hat{x}; k)$ be the far-field patterns of the scattering field $u^s_j$ to the problem \eqref{direct1}--\eqref{direct3} with respect to the scatterer $(D,n_j,\eta_j,\lambda)$ for $j=1,2.$ Assuming that  
    $$|n_1-n_2|^{-1}\in L^{\infty}(D) \quad \text{and} \quad |\eta_1-\eta_2|^{-1}\in L^{\infty}(\partial D)$$ 
    with $\lambda\neq 0$ a fixed constant, then $u_1^{\infty}(\hat{x}; k)\neq u_2^{\infty}(\hat{x}; k)$ for all $\hat{x}\in \mathbb{S}^{d-1}$ and $k\in (k_{min},k_{max})$.
\end{theorem}
\begin{proof}
    We proceed by contradiction and assume that $u_1^{\infty}(\hat{x}; k)= u_2^{\infty}(\hat{x}; k)$ for all $\hat{x}\in \mathbb{S}^{d-1}$ and $k\in (k_{min},k_{max})$. Therefore, by Rellich's Lemma $u_1^s=u^s_2$ for all $x\in\mathbb{R}^d\setminus \overline{D}$. As a consequence we have that  
$$u_{+}^{(1)}=u_{+}^{(2)}, \quad \partial_{\nu}u_{+}^{(1)}=\partial_{\nu}u_{+}^{(2)},\quad \text{which implies that} \quad u_{-}^{(1)}=u_{-}^{(2)} \quad \text{on $\partial D.$}$$ 
Therefore, by the conductivity condition we have that
    $$ \lambda \partial_{\nu}u^{(1)}-\lambda \partial_{\nu}u^{(2)}=(\eta_1-\eta_2)u^{(2)}\hspace{.3cm}\text{on}\hspace{.3cm}\partial D\hspace{.3cm}\text{with}\hspace{.3cm}\lambda\neq 0.$$
    It is clear to see that $(u^{(1)},u^{(2)})\in H^1(D)\times H^1(D)$ satisfies 
\begin{align}
\Delta u^{(1)}+k^2 n_1 u^{(1)} = 0 \quad \text{ and } \quad  \Delta u^{(2)}+k^2n_2u^{(2)} = 0  \hspace{.2cm}& \text{in} \hspace{.2cm}  D \\
u^{(1)}=u^{(2)} \quad \text{ and }  \lambda \partial_{\nu}u^{(1)}-\lambda \partial_{\nu}u^{(2)} = (\eta_1-\eta_2)u^{(2)} \quad   \  \hspace{.2cm}& \text{on} \hspace{.2cm}  \partial D
\end{align}
where $\lambda\neq 0$ for all $k \in (k_{min},k_{max})$.

Now we appeal to our transmission eigenvalue problem \eqref{eigproblemvariable1}--\eqref{eigproblemvariable2} where we have shown that the set of values is at most discrete. Notice, the discreteness of the transmission eigenvalues implies that there exists $k_0\in (k_{min},k_{max})$ such that $k_0$ is not a transmission eigenvalue. This implies that $u^{(1)}=u^{(2)}=0$ in $D$ and by using the boundary conditions of the problem we have that $u^{(1)} = \partial_\nu u^{(1)} = 0 \hspace{.3cm}\text{on}\hspace{.3cm}\partial D.$
Thus by unique continuation for the Helmholtz equation we have that $u_1^s=-u^i$ in  $\mathbb{R}^d\setminus\overline{D}$. This says that $u^i$ is an entire radiating solution to the Helmholtz equation in $\R^d$. This is a contradiction since this would imply that $u^i = 0$ in $\R^d$ which is in opposition to the fact that $|u^i|=1$ in $\R^d$. This proves the claim.
\end{proof}
As a consequence of this result, we have the following uniqueness result.

\begin{corollary}
If $n, \eta \in \mathbb{C}$ are constant, then $u^{\infty}(\hat{x},\hat{y})$  for all $\hat{x}\in \mathbb{S}^{d-1}$ and $k\in (k_{min},k_{max})$ uniquely determines $n$ or $\eta.$
\end{corollary}


We have shown the uniqueness of the refractive index variable $n$ and complex constant $\eta$ for the inverse problem with two conductivity parameters. Now, we will analyze the scattered field $u^s(x)$ with two conductivity parameters and reconstruct some regions. 

\section{Direct Sampling Method}\label{factor-LDSM}
In this section, we study the Direct Sampling Method (DSM) to solve the problem for the reconstruction of isotropic scatterers. First, we analyze the problem using a full aperture and then a partial aperture. For the partial aperture, this means that we are either limiting the number of observation or/and incident directions. In either case, the Lippmann--Schwinger representation of the scattered field \eqref{totalfield} will be used in the analysis. We will derive a new factorization of the far-field operator defined in \eqref{ff-operator} which is one of the main components of our analysis. We will prove that the new proposed imaging function has the property that it decays as the sampling point moves away from the scatterer.
\subsection{Factorization of the Scattered Field with Full Aperture}

We begin by factorizing the far-field operator defined in \eqref{ff-operator} which will allow us to define an imaging function to facilitate the reconstruction of extended regions $D$. Recall, that the far-field operator for $g\in L^2(\mathbb{S}^{d-1})$ is given by 
$$(Fg)(\hat{x})=\int_{\mathbb{S}^{d-1}}u^\infty(\hat{x},\hat{y})g(\hat{y})\,\text{d}s(\hat{y})$$
where $\mathbb{S}^{d-1}$ is the unit sphere/circle. Since it is well-known that the far-field pattern is analytic (see for e.g. \cite{coltonkress}) it is clear that $F$ is a compact operator. 
 The factorization of the far-field operator was initially studied \cite{Liu} for the case when $\eta=0$ and $\lambda=1$ and in \cite{fmconductbc} it was studied for the case where $\lambda=1$ and $\eta\neq 0$.  Now, by the  Lippmann--Schwinger representation of the scattered field given in \eqref{totalfield} we have that  
$$u^\infty(\hat{x},\hat{y})=\int_{D}k^2(n-1)[u^s+u^i] \text{e}^{-\text{i}k\hat{x}\cdot z}\, \text{d}z+\int_{\partial D}\eta[u^i+u^s]\text{e}^{-\text{i}k\hat{x}\cdot z}\, \text{d}s(z)$$
$$ +\int_{\partial D}(1-\lambda)\big[\partial_\nu u^s_{-} +\partial_\nu u^i \big]\text{e}^{-\text{i}k\hat{x}\cdot z}\, \text{d}s(z).$$
Using the above formula for the far-field pattern, we can change the order of integration to obtain the following identity 
$$Fg=\int_{D}k^2(n-1)[u_g^s+v_g] \text{e}^{-\text{i}k\hat{x}\cdot z} \, \text{d}z+\int_{\partial D}\eta[u_{g,-}^s+v_g] \text{e}^{-\text{i}k\hat{x}\cdot z}\, \text{d}s(z)$$
$$+\int_{\partial D}(1-\lambda)[\partial_\nu u_{g,-}^s+\partial_\nu v_g] \text{e}^{-\text{i}k\hat{x}\cdot z}\, \text{d}s(z)$$
Here, we let $v_g(x)$ denote the Herglotz wave function defined as  
$$v_g (x) =\int_{\mathbb{S}^{d-1}} \text{e}^{\text{i}k{x}\cdot \hat{y}}g(\hat{y}) \,  \text{d}s(\hat{y}) \quad \text{ and } \quad u_g^{s} (x)=\int_{\mathbb{S}^{d-1}} u^s(x,\hat{y})g(\hat{y}) \, \text{d}s(\hat{y})$$
where $u_g^s$ solves the boundary value problem \eqref{direct1}--\eqref{direct3} when the incident field $u^i=v_g$. The factorization method for the far-field operator $F$ is based on factorizing $F$ into three distinct pieces that act together and give us more information about the region of interest $D$. To this end, one can show that  
\begin{equation} \label{Hdef}
H:L^2(\mathbb{S}^{d-1})\longrightarrow L^2(D)\times L^2(\partial D) \times H^{-\frac{1}{2}}(\partial D)
\end{equation}
where $ Hg=(v_g|_{D} \, ,\, v_g|_{\partial D}\, ,\, \partial_{\nu}v_g|_{\partial D})^{\top}$
is a bounded linear operator. Now, we consider the following auxiliary problem  
\begin{align}
(\Delta +k^2 n)w=-k^2(n-1)f \hspace{.2cm}& \text{in} \hspace{.2cm}  \mathbb{R}^d\setminus \partial D\label{aux1}\\
w_{-} - w_{+}=0 \quad  \text{and } \quad \partial_{\nu}w_{-} - \partial_\nu w_{+}= -\eta(w+h) -(1-\lambda)(\partial_{\nu}w_{-}+j)  \hspace{.2cm} &\text{on }  \hspace{.2cm} \partial D\label{aux2}
\end{align}
with $f\in L^2(D)$, $h\in L^2(\partial D),$ and $j\in H^{-\frac{1}{2}}(\partial D)$. It is clear that the auxiliary problem \eqref{aux1}--\eqref{aux2} along with the radiation condition is well-posed by \cite{fmconductbc} with $w \in H^1_{loc}(\R^d)$ under the assumptions of this paper. We can define the operator $T$ associated with the auxiliary problem \eqref{aux1}--\eqref{aux2} such that  
$$T:L^2(D)\times L^2(\partial D)\times H^{-\frac{1}{2}}(\partial D)\longrightarrow L^2(D)\times L^2(\partial D)\times H^{-\frac{1}{2}}(\partial D)$$ 
which is given by   
\begin{equation}
    T(f,h,j)^{\top}=\left(k^2(n-1)(w+f)|_{D} \, , \, \eta(w+h)|_{\partial D}\, , \, (1-\lambda)(\partial_{\nu}w_{-}+j)\right)^{\top}. \label{defT}
\end{equation}
Similarly, we have that  $T$ is a bounded and linear operator. Due to the fact that $u_g^s$ solves \eqref{aux1}--\eqref{aux2} with $f=v_g |_{D}$, $h = v_g |_{\partial D}$, and $j=\partial_{\nu}v_g |_{\partial D}$ we have that 
$$THg = \left(k^2(n-1)(u^s_g+v_g)|_{D} \, , \, \eta(u^s_g+v_g)|_{\partial D}\, , \, (1-\lambda)(\partial_{\nu}u_{g,-}^s+\partial_{\nu}v_g)|_{\partial D}\right)^{\top}$$
for any $g \in L^2(\mathbb{S}^{d-1})$. In order to determine a suitable factorization of the far-field operator $F$, we need to compute the adjoint of the operator $H$. It is clear (see for e.g. \cite{dm-1cbclandweber}) that for $(\varphi_1,\varphi_2,\varphi_3)^{\top} \in L^2(D)\times L^2(\partial D)\times H^{\frac{1}{2}}(\partial D)$ that the adjoint operator is given by
$$H^*(\varphi_1,\varphi_2,\varphi_3)^{\top}=\int_D\varphi_1 (z) \text{e}^{-\text{i}k\hat{x}\cdot z}\, \text{d}z+\int_{\partial D}\varphi_2 (z)  \text{e}^{-\text{i}k\hat{x}\cdot z}\, \text{d}s(z)+\int_{\partial D}\varphi_3 (z)  \text{e}^{-\text{i}k\hat{x}\cdot z}\, \text{d}s(z).$$
We then obtain
$$H^*THg=\int_{D}k^2(n-1)[u_g^s+v_g] \text{e}^{-\text{i}k\hat{x}\cdot z} \, \text{d}z+\int_{\partial D}\eta[u_g^s+v_g] \text{e}^{-\text{i}k\hat{x}\cdot z}\, \text{d}s(z)$$
$$+\int_{\partial D}(1-\lambda)[\partial_{\nu}u_{g,-}^s+\partial_{\nu}v_g] \text{e}^{-\text{i}k\hat{x}\cdot z}\, \text{d}s(z)=Fg$$
for any $g \in L^2(\mathbb{S}^{d-1})$. Therefore, we have derived a factorization for the far-field operator. 

\begin{theorem}
The far-field operator $F:L^2(\mathbb{S}^{d-1})\longrightarrow L^2(\mathbb{S}^{d-1}) $ has the symmetric factorization $F=H^*TH$ where the operators $H$ and $T$ are defined in \eqref{Hdef} and \eqref{defT}, respectively. 
\label{Ffact}
\end{theorem}
The factorization given above is one of the main pieces that will be used to derive an imaging functional. The next step in our analysis is the Funk--Hecke integral identity, this integral identity gives us the opportunity to evaluate the Herglotz wave function for $g = \phi_z$ which is given by 
\begin{equation}
v_{\phi_z} (x)= \int_{\mathbb{S}^{d-1}} \text{e}^{-\text{i}k(z-x)\cdot \hat{y}} \text{d}s(\hat{y})=
    \begin{cases}
        2\pi J_0(k|x-z|) & \text{in } \mathbb{R}^2\\
        
        4\pi j_0(k|x-z|) & \text{in } \mathbb{R}^3.
    \end{cases} \label{Herglotz}
\end{equation}
Here, $J_0$ is the zeroth order Bessel function of the first kind and $j_0$ is the zeroth order spherical Bessel function of the first kind. With the factorization of the far-field operator $F$ and the Funk--Hecke integral identity, we can solve the inverse problem of recovering $D$ by using the decay of the Bessel functions (similarly done in \cite{nf-harris,harris-dsm,harris-liem}). For the reconstruction of the region(s) $D$, we use the direct sampling method which implies that we have an imaging functional of the form; 
\begin{equation}
    W_{\text{DSM}}(z)=\left|  \big(\phi_z,F\phi_z\big)_{L^2(\mathbb{S}^{d-1})} \right|. 
    \label{OSMdef}
\end{equation}
To this end, by the definition of $v_g$ and the bound of the operator $T$, we have 
\begin{align*}
   \left|\big(\phi_z,F\phi_z\big)_{L^2(\mathbb{S}^{d-1})}\right|&= \left|\big(\phi_z,H^*TH\phi_z\big)_{L^2(\mathbb{S}^{d-1})}\right|\\
   &= \left|\big(H\phi_z,TH\phi_z\big)_{L^2(D)\times L^2(\partial D)\times H^{-\frac{1}{2}}(\partial D)}\right|\\
    &\leq C\left\|H\phi_z\right\|^2_{L^2(D)\times L^2(\partial D)\times H^{-\frac{1}{2}}(\partial D)}\\
 &=C\left(\|v_{\phi_z}\|^2_{L^2(D)}+\|v_{\phi_z}\|^2_{L^2(\partial D)}+\|\partial_{\nu} v_{\phi_z}\|^2_{H^{-\frac{1}{2}}(\partial D)}\right).
\end{align*} 
Clearly $\|v_{\phi_z}\|_{L^2(D)}$ and $\|v_{\phi_z}\|_{L^2(\partial D)}$ are bounded by the $H^1(D)$ norm of $v_{\phi_z}$ by the continuous embedding of $H^1(D)$ into $L^2(D)$ and the Trace Theorem. For the last term, we use the fact that $v_{\phi z}$ satisfies Helmhotz equation and as a consequence we have 
$$\|\partial_{\nu} v_{\phi_z}\|_{H^{-\frac{1}{2}}(\partial D)}\leq C \left(\|v_{\phi_z}\|_{H^1(D)}+\|\Delta v_{\phi_z}\|_{L^2(D)} \right)\leq C \|v_{\phi_z}\|_{H^1(D)}.$$ 
Thus, the inequality above together with how the Bessel functions decay gives us the following result that will characterized how our imaging function will decay.   

\begin{theorem}
For all $z\in \mathbb{R}^d\setminus \overline{D}$
$$ \left|\big(\phi_z,F\phi_z\big)_{L^2(\mathbb{S}^{d-1})}\right|= \mathcal{O}(\text{dist}(z,D)^{1-d})\hspace{.5cm}\text{for}\hspace{.5cm} \text{dist}(z,D)\to \infty.$$ \label{distD}
\end{theorem}

Here the indicator $W_{\text{DSM}}(z)$ given by \eqref{OSMdef} will decay as $z$ moves away from the region $D$. Here we have used the fact that in $\R^2$ we have that $J_0(t)$ and its derivative decay at a rate of $t^{-1/2}$ as $t\to \infty$ and in $\mathbb{R}^3$ we have $j_0(t)$ and its derivatives decaying at a rate of $t^{-1}$ as $t\to \infty.$ This theorem gives the resolution analysis for using the imaging function as it implies that the imaging function will decay when we move away from the scatterer. 


We also take into consideration the stability of the imaging function \eqref{OSMdef}. We analyze the imaging function $W_{\text{DSM}}(z)$ given by \eqref{OSMdef} with respect to a given/measured perturbed far-field operator. The following theorem addresses the stability of the imaging function.

\begin{theorem}    \label{distDanderror}
    Assume that $\|F^{\delta}-F\|<\delta$ as $\delta\longrightarrow 0$, then for all $z \in \R^d$ 
    $$\left|\left|\big(\phi_z,F^{\delta}\phi_z\big)_{L^2({\mathbb{S}^{d-1}})}\right|- \left|\big(\phi_z,F\phi_z\big)_{L^2({\mathbb{S}^{d-1}})}\right|\right|< \delta  \quad \text{as }\quad \delta\longrightarrow 0.$$
\end{theorem}
\begin{proof}
    Using the triangle inequality we have that 
    \begin{align*}
\left|\left|\big(\phi_z,F^{\delta}\phi_z\big)_{L^2({\mathbb{S}^{d-1}})}\right|- \left|\big(\phi_z,F\phi_z\big)_{L^2({\mathbb{S}^{d-1}})}\right|\right| & \leq \left|\big(\phi_z,(F^{\delta}-F)\phi_z\big)_{L^2({\mathbb{S}^{d-1}})}\right|\\
&\leq \left\| \phi_z\right\|^2_{L^2({\mathbb{S}^{d-1}})} \big\| F^{\delta}-F\big\| = 2^{d-1}\pi  \big\| F^{\delta}-F\big\|
\end{align*}
where on the last line we have used the Cauchy--Schwarz inequality, proving the claim.
\end{proof}
The stability result concludes the analysis for full aperture data. In the following section we present the analysis for the case where we have a limited aperture in $\mathbb{R}^2$.

\subsection{Remarks on Limited Aperture Data in $\mathbb{R}^2$}
We begin this section by assuming that we have a limited number of sources and/or receivers for the partial aperture problem dealing with two conductivity parameters. The sources here come from $\mathbb{S}^{1}_s \subsetneq  \mathbb{S}^1$ such that $\hat{y}\in \mathbb{S}^1_s$ and where 
$$ \mathbb{S}^1_s=\{(\text{cos}\theta,\text{sin}\theta):\theta\in [\alpha_s,\beta_s]\subsetneq [0,2\pi]\}.$$
Similarly, the receivers come from $\mathbb{S}^{1}_m\subsetneq  \mathbb{S}^1$ such that $\hat{x}\in \mathbb{S}^1_m$ and where 
$$\mathbb{S}^1_m=\{(\text{cos}\phi,\text{sin}\phi):\phi\in [\alpha_m,\beta_m]\subsetneq [0,2\pi]\}.$$ Using a similar analysis as before, we have a far-field pattern of the form $u^{\infty}(\hat{x},\hat{y})$ for limited sources and/or receivers. As a consequence of the far-field pattern, one can define the far-field operator with limited aperture as in \eqref{ff-operator}, but now given by 
\begin{equation}\label{ff-limitedoperator}
(\widetilde{F}g)(\hat{x})=\int_{\mathbb{S}^{1}_s}u^{\infty}(\hat{x},\hat{y})g(\hat{y})\, \text{d}s(\hat{y}) \Big|_{\hat{x} \in \mathbb{S}^1_m} \quad \text{ for } \quad g\in L^2(\mathbb{S}^{1}_s)
\end{equation}
where $\widetilde{F}: L^2(\mathbb{S}_s^{1}) \to L^2(\mathbb{S}_m^{1})$. Equivalently one wishes to construct an indicator function \eqref{OSMdef} is valid where one only has the limited far-field operator. Thus, we aim to factorize the operator $\widetilde{F}$ using an analogous analysis as in the previous section. To this end, we define that operator associated with limited aperture as 
\begin{equation*}
H_{s} g=\big(\widetilde{v}_g|_{D} \, ,\, \widetilde{v}_g|_{\partial D}\, ,\, \partial_{\nu}\widetilde{v}_g|_{\partial D}\big)^{\top} \quad \text{where} \quad \widetilde{v}_g (x) =\int_{ \mathbb{S}^1_s} \text{e}^{\text{i}k{x}\cdot \hat{y}}g(\hat{y}) \, \text{d}s(\hat{y})
\end{equation*}
and the adjoint operator for the $H_{m}$ associated the limited receivers is given by 
\begin{equation*}
H_m^*(\varphi_1,\varphi_2,\varphi_3)^{\top}=\int_D\varphi_1 (z) \text{e}^{-\text{i}k\hat{x}\cdot z}\, \text{d}z+\int_{\partial D}\varphi_2 (z)  \text{e}^{-\text{i}k\hat{x}\cdot z}\, \text{d}s(z)+\int_{\partial D}\varphi_3 (z)  \text{e}^{-\text{i}k\hat{x}\cdot z}\, \text{d}s(z)\Big|_{\hat{x} \in \mathbb{S}^1_m} 
\end{equation*}
where $H_m^*$ can be easily computed as in the previous section. With this, we have the following theorem addressing the factorization for the operator $\widetilde{F}.$

\begin{theorem}
The far-field operator $\widetilde{F}:L^2(\mathbb{S}_s^{1})\longrightarrow L^2(\mathbb{S}_m^{1}) $ that considers the limited aperture problem for sources and/or receivers has the symmetric factorization $\widetilde{F}=H_m^*TH_s$ where $T$ is as defined in \eqref{defT}. 
\end{theorem}

Now that we have the factorization of the operator $\widetilde{F}$, we want to formulate an indicator function that will accurately recover the scatterer(s) as the indicator function \eqref{OSMdef}. We first consider the expansion of an integral that was first analyzed in \cite{limitedaperture} for a limited aperture problem. We hope that by considering such form, we can construct an indicator function and show that it will decay as we move away from the scatterer. We begin by investigating the integral expansion presented in \cite{limitedaperture} that considers limited sources and/or receivers
\begin{align}
    \int_{\mathbb{S}_p^1} \text{e}^{\text{i}kx\cdot \hat{y}}\text{ds}(\hat{y})&=(\alpha_p-\beta_p)J_0(k|x|) \nonumber \\ 
    &+4\sum_{\ell=1}^{\infty}\frac{\text{i}^\ell}{\ell}J_\ell(k|x|)\text{cos}\left(\frac{\ell(\alpha_p+\beta_p-2\phi)}{2}\right)\text{sin}\left(\frac{\ell(\alpha_p-\beta_p)}{2}\right)    \label{expansionofsum}
        \end{align} 
where $\phi$ is the polar angle of $x$ and $p=m,s$.

We will use \eqref{expansionofsum} to numerically show that 
$$ \left\|H_p \phi_z \right \|_{L^2(D)\times L^2(\partial D)\times H^{-\frac{1}{2}}(\partial D)} =\mathcal{O}\big(\text{dist}(z,D)^{-1/2}\big)\hspace{.5cm}\text{for}\hspace{.5cm} \text{dist}(z,D)\to \infty.$$
First, it is well known that Bessel functions $J_\ell (t)$ have the following asymptotic behavior 
$$J_\ell (t)\sim \Big(\frac{\text{e}t}{2\ell}\Big)^\ell \quad \text{as } \quad \ell\to \infty$$ 
which implies that \eqref{expansionofsum} converges uniformly on any compact subset of $\R^2$.  Similarly, one can take the partial derivative with respect to either $|x|=r$ or $\phi$, and use the recursive relationship $J'_\ell(t)=\frac{1}{2}\Big(J_{\ell-1}(t)-J_{\ell+1}(t)\Big)$ to show that the partial derivatives converge uniformly on any compact subset of $\R^2$. Now, we provide a numerical example to show the decay of \eqref{expansionofsum} as $|x| \to \infty.$ Therefore, in the following example we fix 
$$\alpha_p=\frac{3\pi}{2},\quad \beta_p=\frac{\pi}{4}\quad  \text{and}\quad  \phi=\frac{\pi}{2}$$ 
in \eqref{expansionofsum} plot the absolute values of the truncated sum such that $\ell=1 , \ldots , 15$. In Figure \ref{sumplot}, we see that truncated function in \eqref{expansionofsum} and it's partial derivative decays like the function $|x|^{-1/2}$ as $|x| \to \infty$ just as in the case of full aperture data. 

\begin{figure}[ht]
\centering 
\includegraphics[width=7cm]{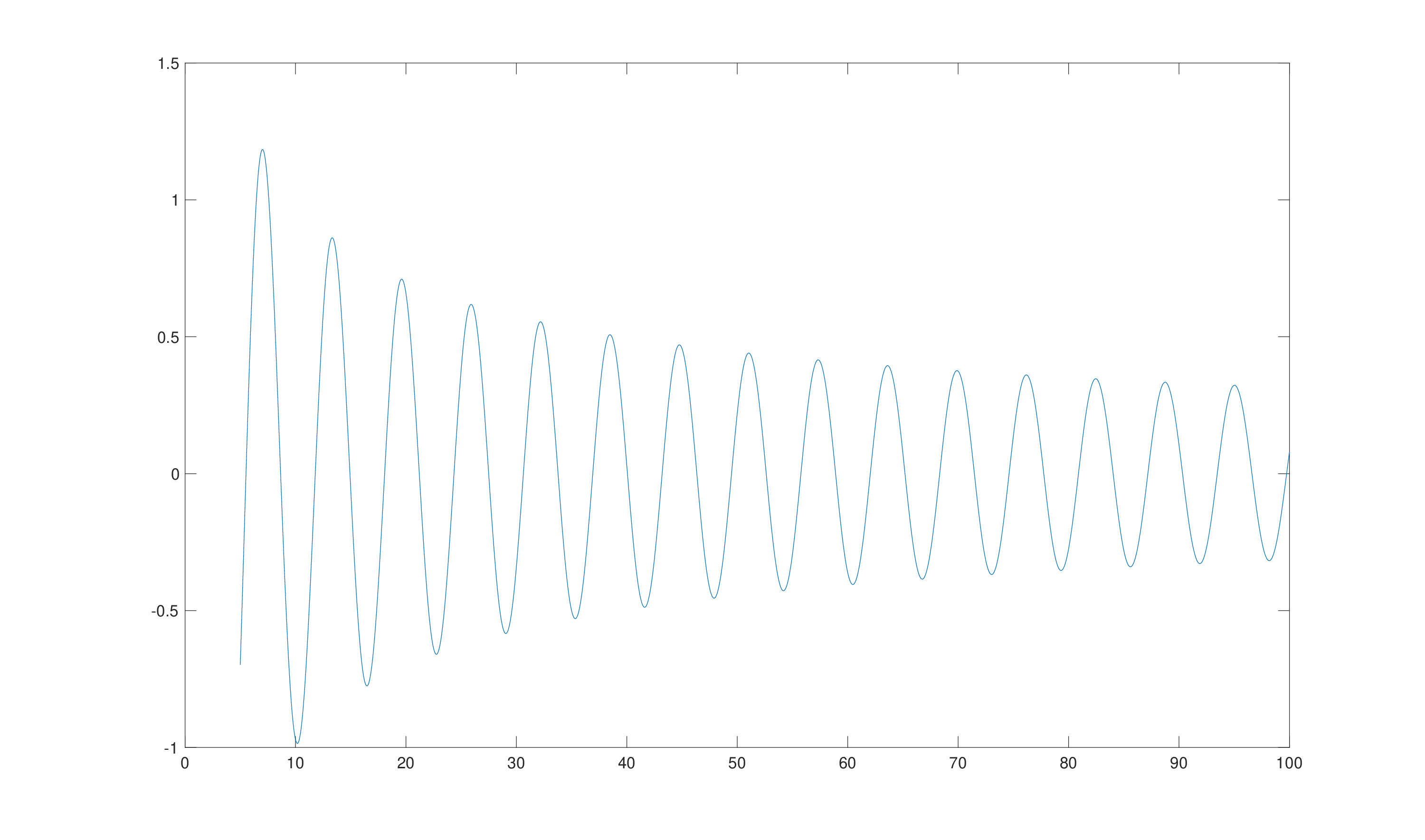}\hspace{.5cm}\includegraphics[width=7cm]{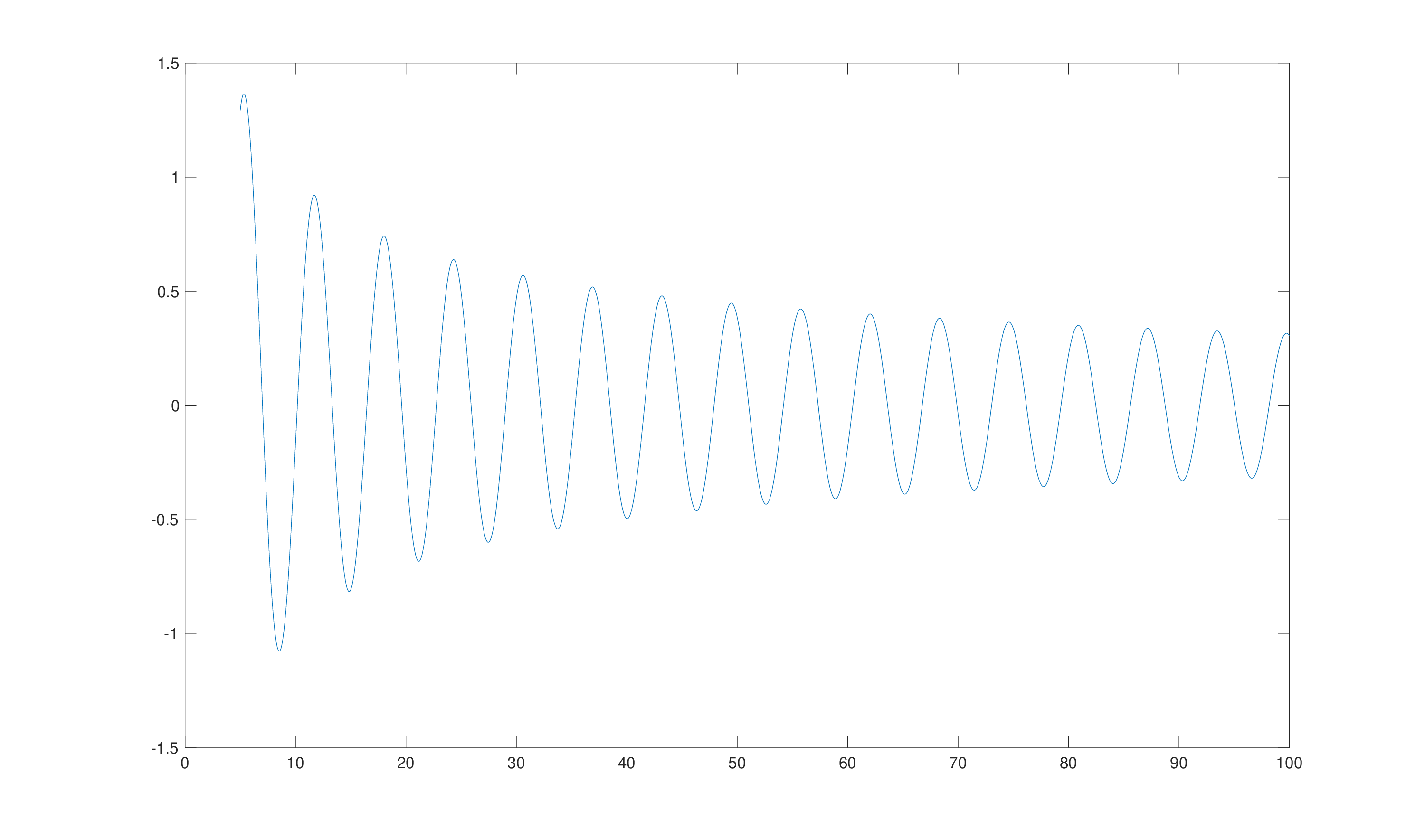} 
\caption{Left to right: plot of \eqref{expansionofsum} up to $s=15$ and plot of the partial derivative of \eqref{expansionofsum} with respect to $r=|x|$ for $s=15$.} 
\label{sumplot}
\end{figure}

The observations on the behavior of \eqref{expansionofsum} allows one to make a numerical conjecture about the decay of the sum that will characterize the imaging function for the case when one only has the limited aperture far-field operator associated with the scattering problem. 

\begin{remark}
Let $\widetilde{F}$ be the limited aperture far-field operator. Then for all $z\in \mathbb{R}^2\setminus \overline{D}$
\begin{align*}
\left|\big(\phi_z,\widetilde{F}\phi_z\big)_{L^2(\mathbb{S}^{1}_m)}\right|&=
\left|\big(H_m\phi_z,TH_s\phi_z\big)_{L^2(D)\times L^2(\partial D)\times H^{-\frac{1}{2}}(\partial D)}\right|\\
 &\leq C\left\|H_m\phi_z\right\|_{L^2(D)\times L^2(\partial D)\times H^{-\frac{1}{2}}(\partial D)}\left\|H_s\phi_z\right\|_{L^2(D)\times L^2(\partial D)\times H^{-\frac{1}{2}}(\partial D)}\\
 &=\mathcal{O}(\text{dist}(z,D)^{-1})\hspace{.5cm}\text{for}\hspace{.5cm} \text{dist}(z,D)\to \infty.
\end{align*} 
\end{remark}
Here we can construct an indicator similar as \eqref{OSMdef} that will decay as $z$ moves away from the region $D$ when $F$ is replaced with $\widetilde{F}$. We will show that we can still recover the scatterer $D$ with limited aperture data. Now, we are ready to present numerical examples for the full and limited aperture problems.

\section{Numerical Validation}
\subsection{Boundary Integral Equations}
We first derive the boundary integral equation to compute far-field data for arbitrary domains in two dimensions which are defined through a smooth parametrization. For simplicity we will assume that the coefficients $n,\eta$ and $\lambda$ are all constant. To begin, since that total field $u$ and radiating scattered field $u^s$ satisfies 
$$\Delta u^s +k^2 u^s=0  \quad \textrm{ in }  \R^2 \setminus \overline{D}  \quad  \text{and} \quad  \Delta u +k^2 nu=0  \quad \textrm{ in } \,  {D} $$
we make the ansatz that 
$$u = (\mathrm{SL}_{k\sqrt{n}} \psi_1)(x)  \quad \textrm{ in } \,  {D}   \quad \text{and} \quad u^s(x) = (\mathrm{SL}_{k} \psi_2)(x) \quad \textrm{ in }  \R^2 \setminus \overline{D}.$$ 
Here for the wave number $\tau$ we have that the single layer potential 
$$ (\mathrm{SL}_{\tau} \psi )(x) = \int_{\partial D} \Phi_{\tau}(x,w)\psi(w)\,\mathrm{d}s(w)$$
where $\Phi_{\tau}(x,y)$ is the fundamental solution to the Helmholtz equation in $\R^2$. Recall that the single layer potential satisfies the equation in $D$ for $\tau =k\sqrt{n}$  and $\R^2 \setminus \overline{D}$ for $\tau =k$ along with the radiation condition for $u^s$. Therefore, we need to find $(\psi_1 ,\psi_2)^{\top}$ that satisfies the boundary conditions 
$$ (u^s+u^i )^+ - u^-=0  \quad  \text{and} \quad \partial_\nu (u^s+u^i )^+ + \eta (u^s+u^i )^+= \lambda {\partial_{\nu} u^-} \quad \textrm{ on } \,  \partial D $$
where the incident plane wave is given by form $u^i=\text{e}^{\text{i} k x \cdot \hat{y} }$ with incident direction $\hat{y}\in \mathbb{S}^{1}$.

By the jump relations (see for e.g. \cite{int-equ-book}) we have that the above boundary conditions can be written as 
$$ \mathrm{S}_{k\sqrt{n}}\psi_1-\mathrm{S}_{k}\psi_2=u^i \quad \text{and} \quad \lambda\left(\frac{1}{2}\mathrm{I}+\mathrm{K}'_{k\sqrt{n}}\right)\psi_1-\left(-\frac{1}{2}\mathrm{I}+\mathrm{K}'_{k}\right)\psi_2-\eta \mathrm{S}_{k}\psi_2=\partial_\nu u^i+\eta u^i$$
for all $x \in \partial D$. Here for the wave number $\tau$ we define 
$$ (\mathrm{S}_{\tau} \psi)(x)=\int_{\partial D}\Phi_\tau (x,w)\psi(w)\,\mathrm{d}s(w) \quad  \text{and} \quad (\mathrm{K}'_{\tau}\psi) (x)=\int_{\partial D}\partial_{\nu(x)}\Phi_\tau (x,w)\psi(w)\,\mathrm{d}s(w) $$
for all $x\in \partial D$. This gives a 2 $\times$ 2 system of equation to solve for $(\psi_1 ,\psi_2)^{\top}$. We can then use a standard discretization to solve the above system which implies that the far-field data is given by 
$$ u^\infty (\hat{x} , \hat{y} ) = \int_{\partial D} \text{e}^{-\text{i} k \hat{x} \cdot w} \psi_2(w ; \hat{y})\,\mathrm{d}s(w)$$
which we can use in our numerical experiments. 

In order to check the accuracy of our system of boundary integral equation we compare with the analytical solution when the scatterer $D$ is the unit disk. To this end, we will derive an analytical formula for the far-field pattern $ u^\infty (\hat{x} , \hat{y} )$ via separation of variables. Similar calculations have been in greater details in \cite{dm-1cbclandweber}. With this, we first note that 
$$  u^s(x)=\sum_{p=-\infty}^\infty \mathrm{i}^p a_pH_p^{(1)}(k|x|)\mathrm{e}^{\mathrm{i}p(\theta -\phi)} \quad \text{ and } \quad  u(x)=\sum_{p=-\infty}^\infty \mathrm{i}^p b_pJ_p(k\sqrt{n}|x|)\mathrm{e}^{\mathrm{i}p(\theta -\phi)}$$
where $\theta$ is the angular coordinate for $x$ and $\phi$ is the angular coordinate for $\hat{y}$.
These series solutions satisfy the differential equation and we only need to apply the boundary conditions. Therefore, by the Jacobi-Anger expansion for plane waves, we have that   
$$ \begin{pmatrix} 
  H_p^{(1)}(k) & -J_p(k\sqrt{n})\\
  kH_p^{(1)'}(k)+\eta H_p^{(1)}(k) & -\lambda k\sqrt{n}J_p'(k\sqrt{n})
\end{pmatrix} 
\begin{pmatrix} 
  a_p\\
b_p
\end{pmatrix} 
=\begin{pmatrix} 
  -J_p(k)\\
  -kJ_p'(k)-\eta J_p(k)
\end{pmatrix} 
$$
where the system of equations comes from forcing the boundary conditions. This implies that the `Fourier' coefficients $a_p$ are given by
\begin{eqnarray*}
 a_p=-\frac{
 \lambda k\sqrt{n}J_p(k)J_p' (k\sqrt{n})-J_p(k\sqrt{n})\left(kJ_p'(k)+\eta J_p(k)\right)
 }{
 \lambda k\sqrt{n}H_p^{(1)}(k)J_p'(k\sqrt{n})-J_p(k\sqrt{n})\left(kH_p^{(1)'}(k)+\eta H_p^{(1)}(k)\right)
 }\,.
\end{eqnarray*}
The gives that the far-field pattern has the form 
\begin{eqnarray}
 u^\infty(\hat{x}, \hat{y})=\frac{4}{\mathrm{i}}\sum_{p=-\infty}^\infty a_p \mathrm{e}^{\mathrm{i}p(\theta -\phi)}
  \label{farfieldseries}
\end{eqnarray}
which gives the analytic expression needed.

We let ${\bf F}_k\in \mathbb{C}^{64 \times 64}$ be the matrix containing the far-field data for $64$ equidistant incident directions and $64$ evaluation points for the unit disk with the physical parameters, $\eta$, $\lambda$, $n$ and given wave number $k$ obtained by (\ref{farfieldseries}). We denote by ${\bf F}_k^{(N_f)}$ the far-field data obtained through the boundary element collocation method, where $N_f$ denotes the number of faces in the method. Note that the number of collocation nodes is $3\cdotp N_f$ and the absolute error is defined by 
$$\varepsilon_k^{(N_f)} :=\|{\bf F}_k-{\bf F}_k^{(N_f)}\|_2 .$$
In Table \ref{tablefarfield}, we show the absolute error of the far-field for 64 incident directions and 64 evaluation point, for a disk with radius $R=1$ and the parameters $\eta=2+\mathrm{i}$, $\lambda=2$, and $n=4+\mathrm{i}$ and the wave numbers $k=2$, $k=4$, and $k=6$. As we can observe, we obtain very accurate results using $120$ collocation nodes.

\begin{table}[!ht]
\centering
 \begin{tabular}{c|c|c|c|}
  $N_f$ & $\varepsilon_2^{(N_f)}$ & $\varepsilon_4^{(N_f)}$ & $\varepsilon_6^{(N_f)}$ \\
  \hline 
20 & 0.09931 & 0.72936 & 2.76920\\
40 & 0.00889 & 0.00988 & 0.04348\\
60 & 0.00091 & 0.00077 & 0.00149\\
80 & 0.00010 & 0.00008 & 0.00008\\
\hline
 \end{tabular}
 \caption{\label{tablefarfield}Absolute error of the far-field with 64 equidistant incident directions and 64 evaluation for the disk with $R=1.$}
\end{table}
\subsection{Numerical Examples}
In this section, we present numerical examples for reconstructing extended scatterers via the direct sampling method. We will show that Theorem \ref{distD} can be used to numerically recover the scatterer. Similarly to the previous section, the numerical computations are done in MATLAB 2021b. We used the system of boundary integral equations from the previous section to approximate the discretized far-field operator 
$$\textbf{F}=\Big[u^\infty(\hat{x}_i,\hat{y}_j)\Big]^N_{i,j=1}.$$
We can discretize such that   
$$\hat{x}_i=\hat{y}_i=(\cos(\theta_i),\sin(\theta_i))\quad \text{ where } \quad \theta_i=2\pi(i-1)/64 \,\, \text{ for} \,\,i=1,\dots,64.$$ 
We get then $\textbf{F}$ which is a $64 \times 64$ complex valued matrix with $64$ incident and observation directions. An additional component needed is the vector $\phi_z$ which we compute by
$$\boldsymbol{\phi}_z=\big(\text{e}^{-\text{i}k\hat{x}_1\cdot z},\dots, \text{e}^{-\text{i}k\hat{x}_{64} \cdot z}\big)^\top\hspace{.5cm}\text{where}\hspace{.5cm}z\in \mathbb{R}^2.$$  In order to model experimental error in the data we add random noise to the discretized far-field operator $\mathbf{F}$ such that
$$\mathbf{F}^{\delta}=\Big[ \mathbf{F}_{i,j}(1+\delta\mathbf{E}_{i,j})\Big]^{64}_{i,j=1}\hspace{.5cm}\text{where}\hspace{.5cm} \| \mathbf{E}\|_2=1.$$
Here, the matrix $\mathbf{E} \in \C^{64 \times 64}$ is taken to have random entries and $0<\delta \ll 1$ is the relative noise level added to the data. This gives that the relative error is given by $\delta$.

Thus, numerically we can approximate the imaging function by 

$$W_{\text{DSM}}(z)=\Big|\big(\boldsymbol{\phi}_z,\mathbf{F}^{\delta}\boldsymbol{\phi}_z\big)_{\ell^2} \Big|^2$$
where we use the 2nd power to increase the resolution in the numerical examples.  We consider the following three domains with same fixed parameters: kite, peanut, and star. Their respective parameterizations are given by 
$$\big(-1.5\sin(\theta),\cos(\theta)+0.65\cos(2\theta)-0.65\big)^\top, \quad 2\sqrt{\frac{\sin(\theta)^2}{2}+\frac{\cos(\theta)^2}{10}}\big(\cos(\theta),\sin(\theta)\big)^\top$$
$$\hspace{3cm}\text{and}\quad \frac{1}{4}\Big(2+0.3\cos(5\theta)\Big)\big(\cos(\theta),\sin(\theta)\big)^\top.$$
To analyze the change of different parameters, we use the unit disk as the domain of interest given by the parameterization $\partial D=\big(\cos(\theta),\sin(\theta)\big)^\top$. 

To discretize the problem for limited aperture, we take the same form for $\hat{x}_i$, $\hat{y}_i$, and $\theta_i$ as above for $i=1,\dots,64$. To simulate limited aperture for either the sources and/or receivers, it is enough to either remove columns and/or rows, respectively, from $\textbf{F}$. In this case, we limit the $i$'s for either $\hat{x}_i$ and/or $\hat{y}_i$. For example, if we want a problem with only half of the sources and all receivers, we take $\hat{y}_i$ for $i=1,\dots,32$ and $\hat{x}_i$ for $i=1,\dots,64.$ Now, we are ready to present the numerical results for the full aperture and limited aperture. 
\vspace{.5cm}

\noindent\textbf{Example 1. Recovering a star region with full aperture:}\\
\noindent For the star shaped domain, we assume that the refractive index is $n=4+\text{i}$ and boundary parameters $\eta=2+\text{i}$ and $\lambda=2.$ We will take $k=2\pi$ as the wave number, $\delta=0.10$ which corresponds to the $10\%$ random noise added to the data, and use the imaging functional \eqref{OSMdef} to recover the scatterer.  
\begin{figure}[htb!]
\centering 
\includegraphics[width=12cm]{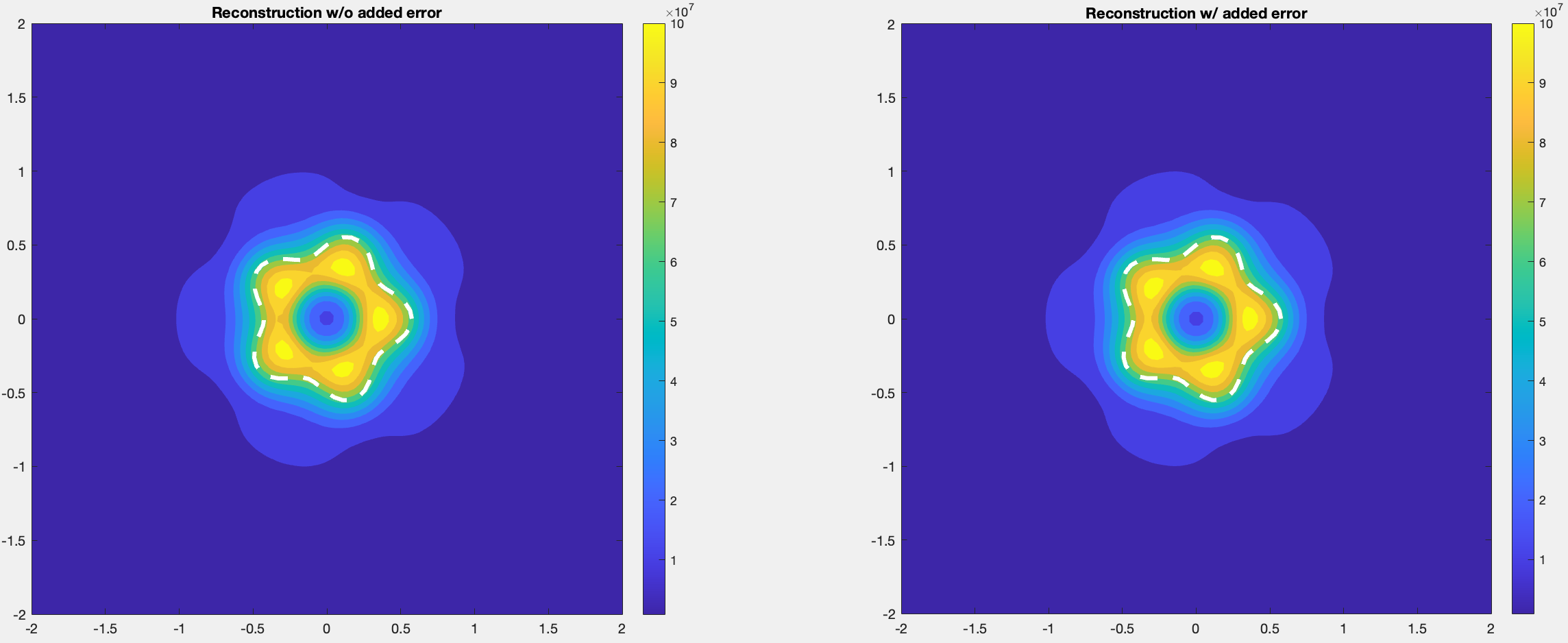}
\caption{Reconstruction star region by the DSM without noise and with 10\% noise.}
\label{star10}
\end{figure}

In Figure \ref{star10}, we see that both images are very similar and both give a good approximation of the star scatterer. Even with noise, we see that we capture the entire scatterer which validates the theory presented in Theorem \ref{distD} and in Theorem \ref{distDanderror}.\\

\noindent\textbf{Example 2. Recovering a peanut region with full aperture:}\\
\noindent For this reconstruction, we take the same values for the physical parameters as example 1. We add 5\% random noise to the data and use the imaging functional presented in \eqref{OSMdef} to recover the region of interest.
\begin{figure}[htb!]
\centering 
\includegraphics[width=12cm]{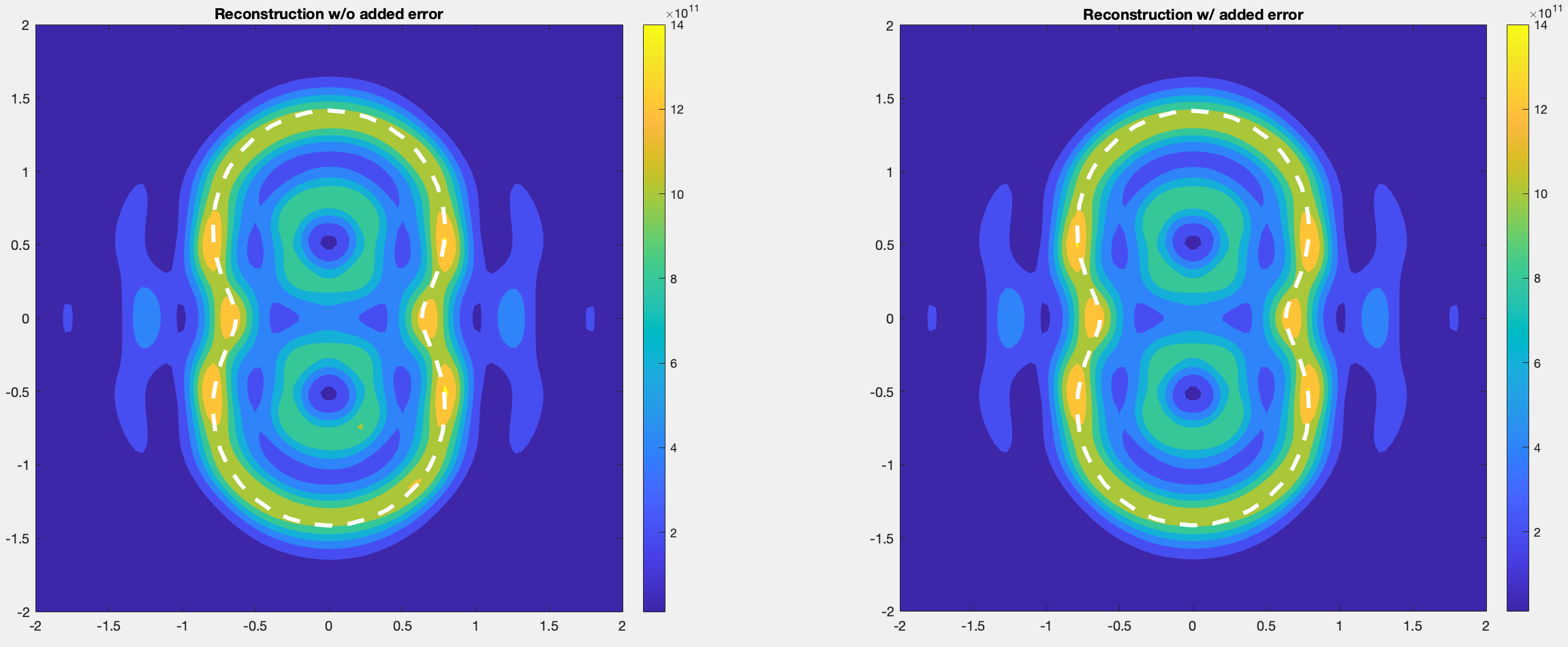}
\caption{Reconstruction of peanut region by the DSM without noise and with  5\% noise.}
\label{peanut5}
\end{figure}

In Figure \ref{peanut5}, we see that with even more random noise added, the reconstruction is very good. Using the imaging functional \eqref{OSMdef} we recover the scatterer in terms of position, location, and size. \\

\noindent\textbf{Example 3. Recovering a kite region with full aperture:}\\
For this numerical experiment, we take the same values for the physical parameters as previous examples and we add 15\% random noise to the data. 
\begin{figure}[htb!]
\centering 
\includegraphics[width=12cm]{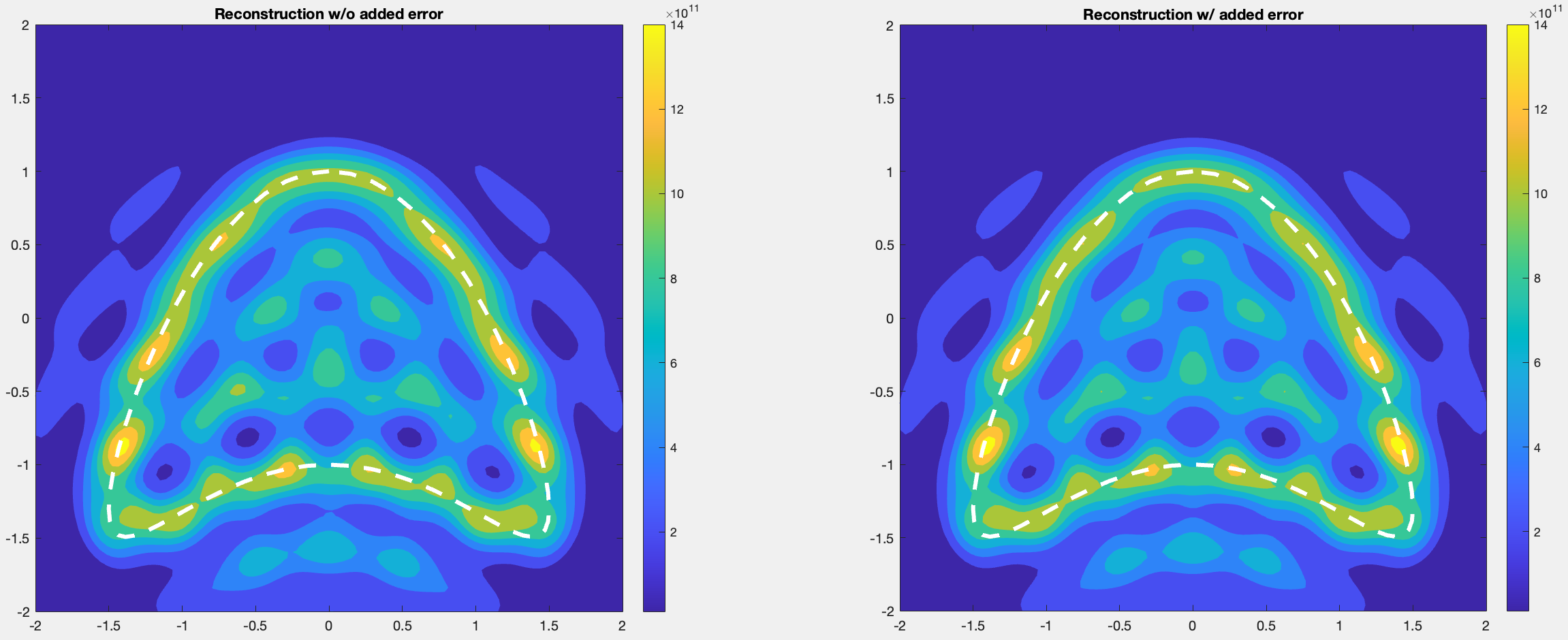}
\caption{Reconstruction of kite region by the DSM without noise and with 15\% noise.}
\label{kite15}
\end{figure}

In this example, we also see that Figure \ref{kite15} does not change when we add noise to the data. Thus, our imaging functional is performing well as we sample the points. We capture entirely the scatterer in terms of position, size, and shape. \\

\noindent\textbf{Example 4. Recovering a kite region with partial aperture (limited receivers):}\\
In the next example, we consider limited aperture data for the kite shape domain. We fix the physical parameters as in previous examples. We first consider the receivers $\hat{x}_i$ coming from fixing $i=1,\dots,32$ and the sources $\hat{y}_i$ coming from fixing $i=1,\dots,64$. In other words, sampling on the receivers only on half of the unit disk, but taking all the sources on the full unit disk. We add 15\% random noise to the data and use an imaging functional similar as \eqref{OSMdef} but making the proper adjustments to account for the limited sources.
\begin{figure}[htb!]
\centering 
\includegraphics[width=12cm]{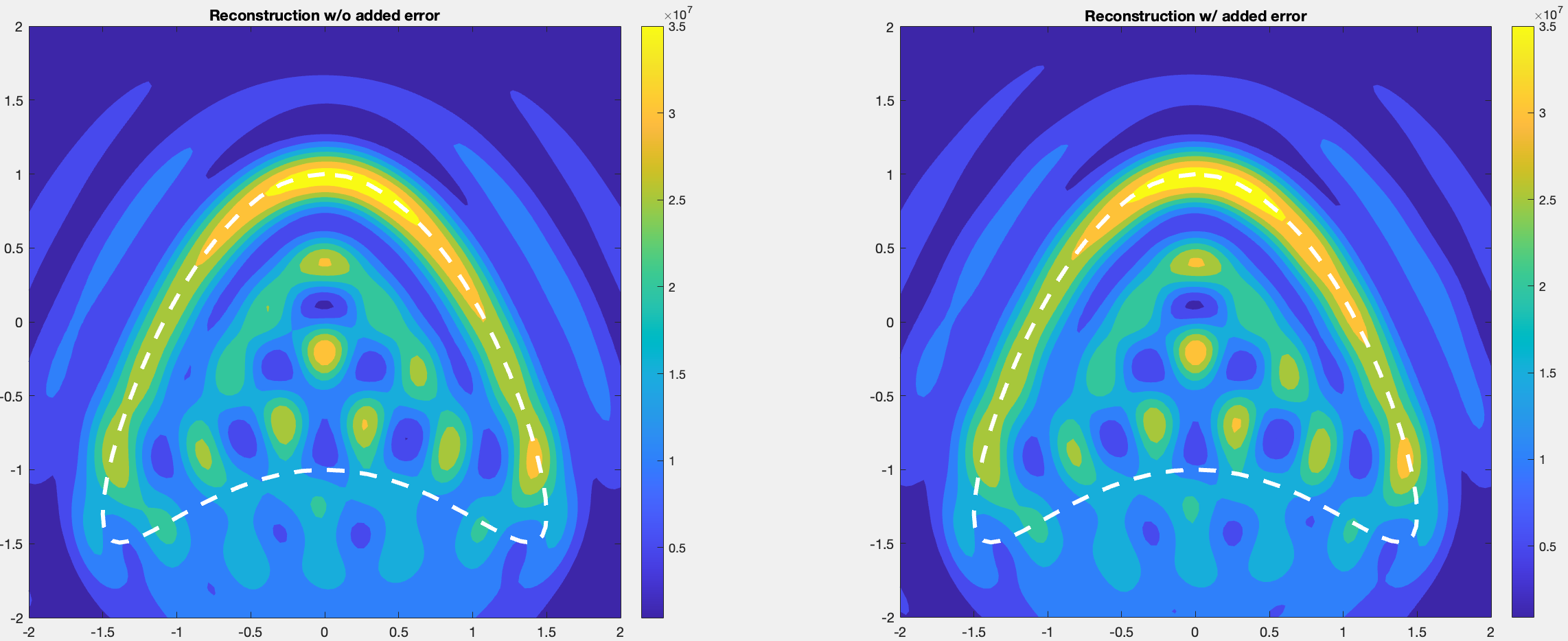}
\caption{Reconstruction of kite region using limited receivers by the DSM without noise and with 15\% noise.}
\label{kite15limited}
\end{figure}

 In Figure \ref{kite15limited}, we see that even though that we had a limited amount of receivers on the unit circle our reconstruction is positive and favorable. We understand the location, the shape, and the size of the kite scatterer even without the full knowledge of the receivers. In addition, we see that adding noise does not change the image and as a consequence we will consider other scatterers with limited data without adding noise.\\

 \noindent\textbf{Example 5. Recovering a kite region with partial aperture using limited sources and limited receivers and sources:}\\
 Let the physical parameters be as in the previous examples and we add $0\%$ noise to the data. In this example, we first limit the amount of receivers and then we limit the amount of receivers and sources. In the left image we consider the limited aperture example that takes the receivers $\hat{x}_i$ coming from fixing $i=1,\dots,64$ and the sources $\hat{y}_i$ coming from fixing $i=1,\dots,48$. In other words, sampling on the receivers for a full unit disk, but taking the first three halves of the unit disk for the sources. In the image on the right, we take the receivers $\hat{x}_i$ coming from fixing $i=1,\dots,32$ and the sources $\hat{y}_i$ coming from fixing $i=16,\dots,64$. So, the receivers come from the first half of the unit disk and the sources come from the last three halves of the unit disk. Thus, we have the following numerical examples for the kite scatterer with limited data. 
\begin{figure}[htb!]
\centering 
\includegraphics[width=12cm]{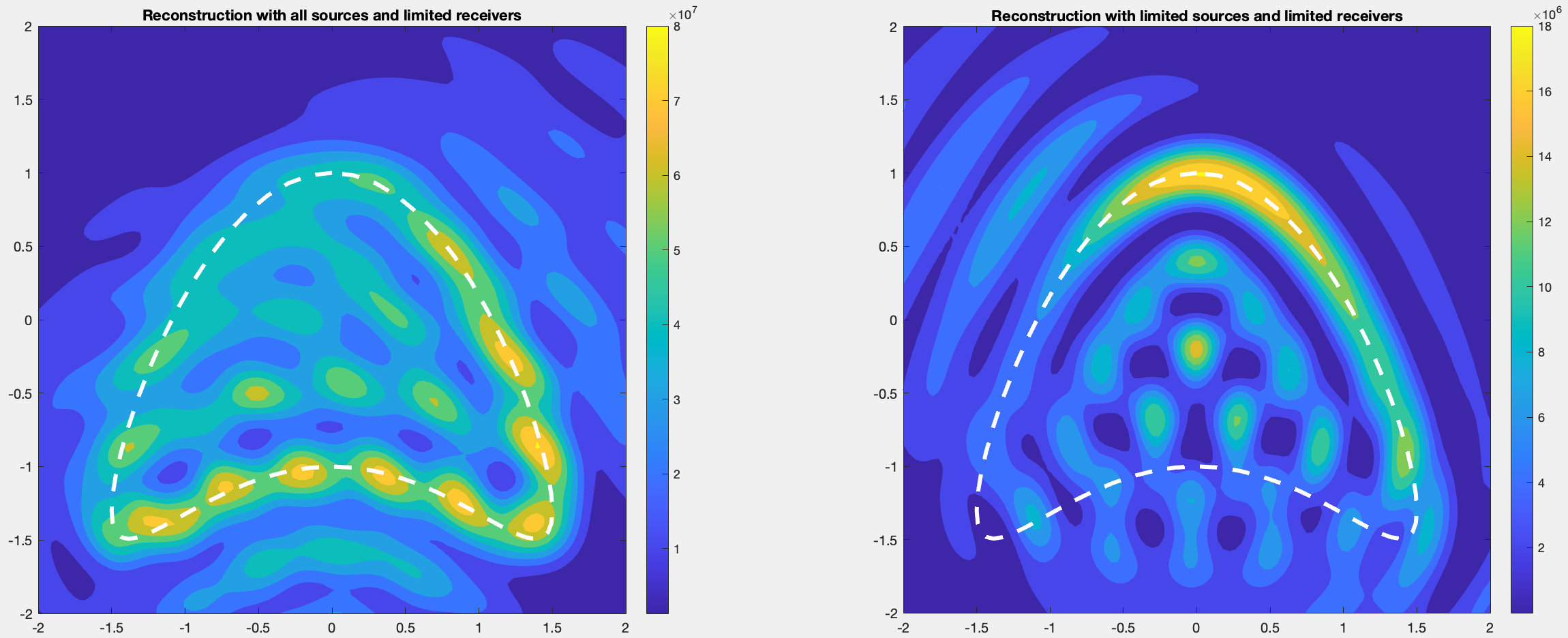}
\caption{Reconstruction of kite region using limited aperture with $0\%$ added by the DSM. Image left to right: reconstruction on limited aperture on the sources and reconstruction on limited aperture on both the receivers and sources.}
\label{kitelimited}
\end{figure}

 In Figure \ref{kitelimited}, we see that having limited information on the receivers and sources simultaneously gives a poor reconstruction in comparison to only limiting the receivers (by previous example) or sources. Having either limited receivers or limited sources still gives us a favorable and positive reconstruction in terms of the location, shape, and size of the kite shape domain which validates our conjecture in the remarks on limited aperture in $\mathbb{R}^2$.\\

\noindent\textbf{Example 6. Recovering a disk region with full aperture:}\\ 
For this reconstruction, we take the same values for the physical parameters as example 1. We add 20\% random noise to the data and use our imaging functional presented in \eqref{OSMdef} to recover the region of interests.
\begin{figure}[htb!]
\centering 
\includegraphics[width=12cm]{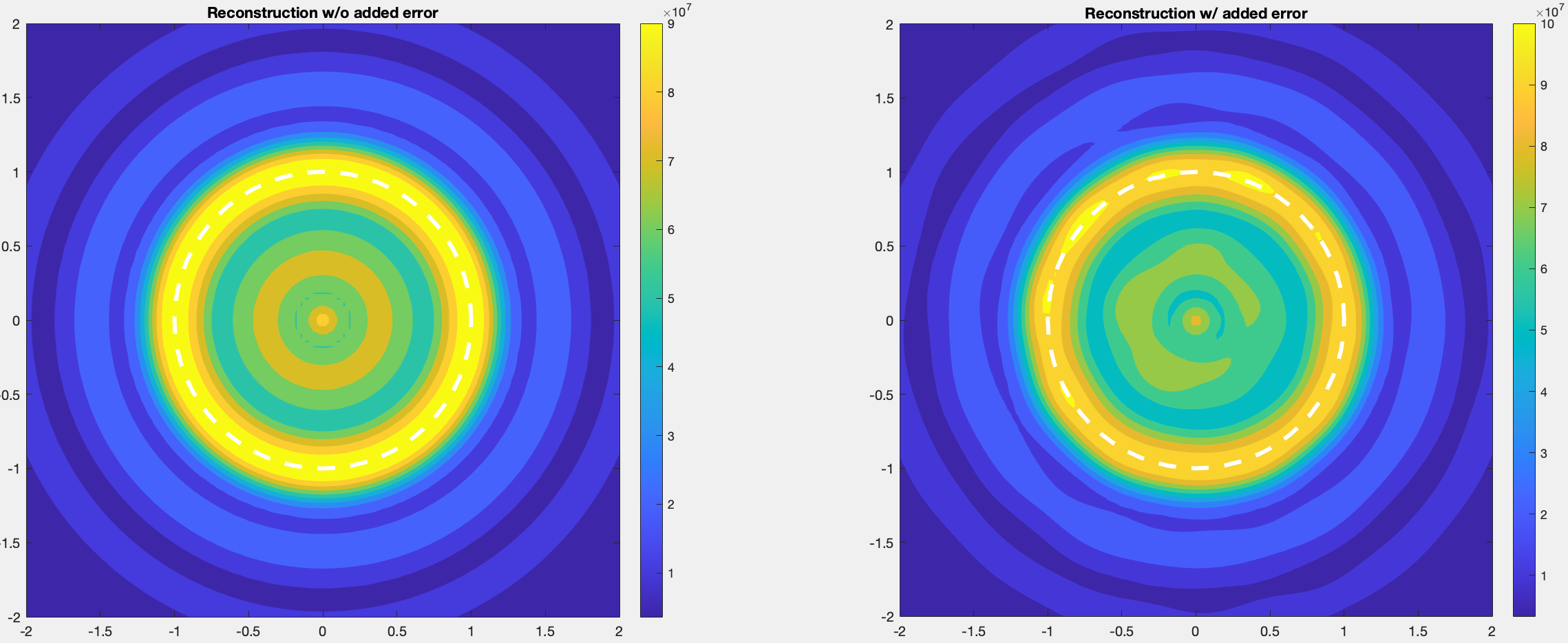}
\caption{Reconstruction of disk region using the DSM without noise and with 20\% noise.}
\label{circle20}
\end{figure}

In this example, we also see that the reconstruction without noise added to the data outperforms the one with noise. However, the one with noise still gives us the scatterer's position, shape, and location. In Figure \ref{circle20}, we see that both images recover the  disk even though the clarity of the one without noise is better. Thus, our imaging functional is performing well even when we add 20\% noise to the data. This validates Theorem \ref{distD} and Theorem \ref{distDanderror}.\\

 \noindent\textbf{Example 7. Recovering a disk region with full aperture using different physical parameters:}\\ 
 For the disk-shaped domain in this example, we assume that the refractive index is $n=4+3\text{i}$ and the boundary parameters $\eta=2+2\text{i}$ and $\lambda=5.$ Here, we will take $k=2\pi$ as the wave number and we let $\delta=0.10$ which corresponds to the $10\%$ random noise added to the data.  
\begin{figure}[htb!]
\centering 
\includegraphics[width=12cm]{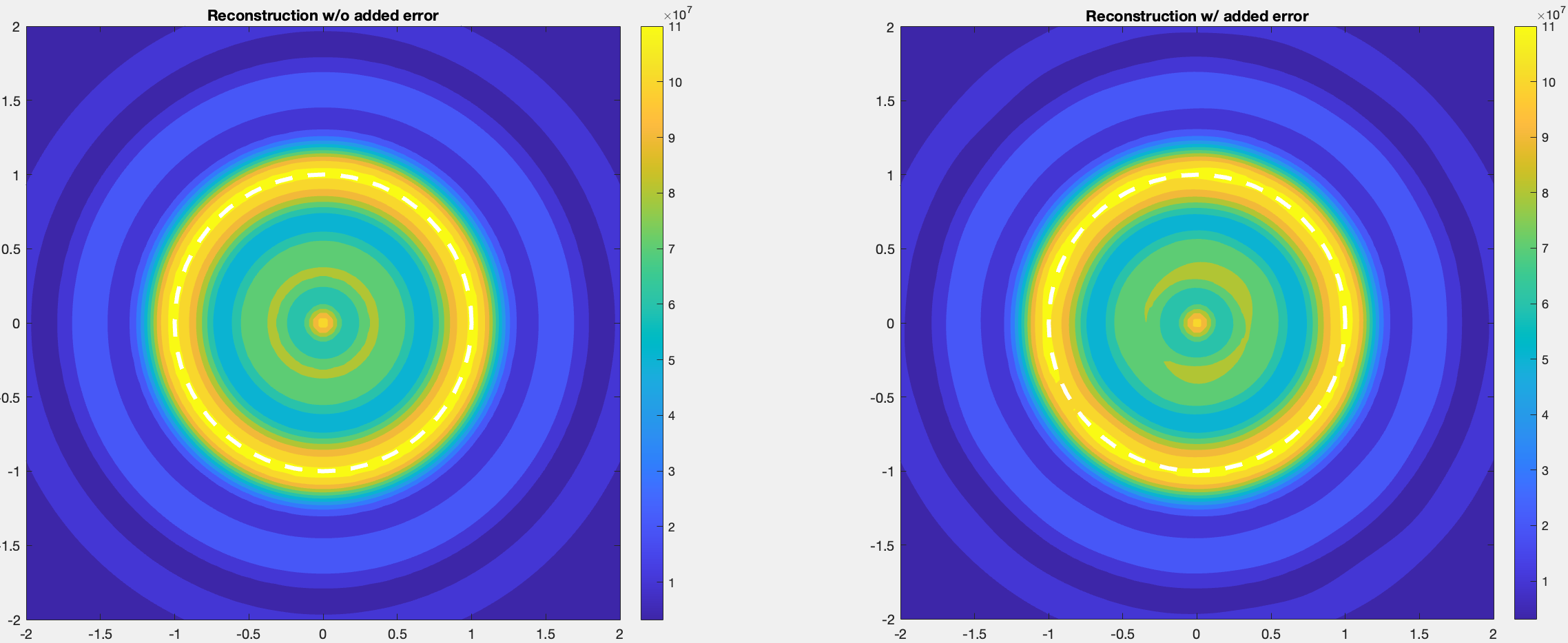}
\caption{Reconstruction of disk region by the DSM without noise and with 10\% noise.}
\label{circlediffpa10}
\end{figure}

 In Figure \ref{circlediffpa10}, we see that even when we change the physical parameters that correspond to the problem, we obtain a positive and favorable reconstruction of the scattering object even with noise added to the data. Adding noise to the data or not, we still recover the disk scatterer in terms of position, size, and location.\\ 

 \noindent\textbf{Example 8. Recovering a disk region with limited aperture:}\\
 Let the physical parameters be as example 7 and here we add $10\%$ noise to the data. In this example, we first limit the amount of receivers and then we limit the amount of sources. In the left image we consider the limited aperture example that takes the receivers $\hat{x}_i$ coming from fixing $i=1,\dots,48$ and the sources $\hat{y}_i$ coming from fixing $i=1,\dots,64$. In other words, sampling on the receivers for the first three halves of the unit disk, but taking the full unit disk for the sources. In the image on the right, we take the receivers $\hat{x}_i$ coming from fixing $i=1,\dots,64$ and the sources $\hat{y}_i$ coming from fixing $i=16,\dots,64$. So, the receivers come from the whole unit disk and the sources come from the last three halves of the unit disk. In both images, we add $10\%$ random noise to the data and thus we have the following numerical examples for the limited aperture for the disk scatterer. 
 
 \begin{figure}[htb!]
\centering 
\includegraphics[width=12cm]{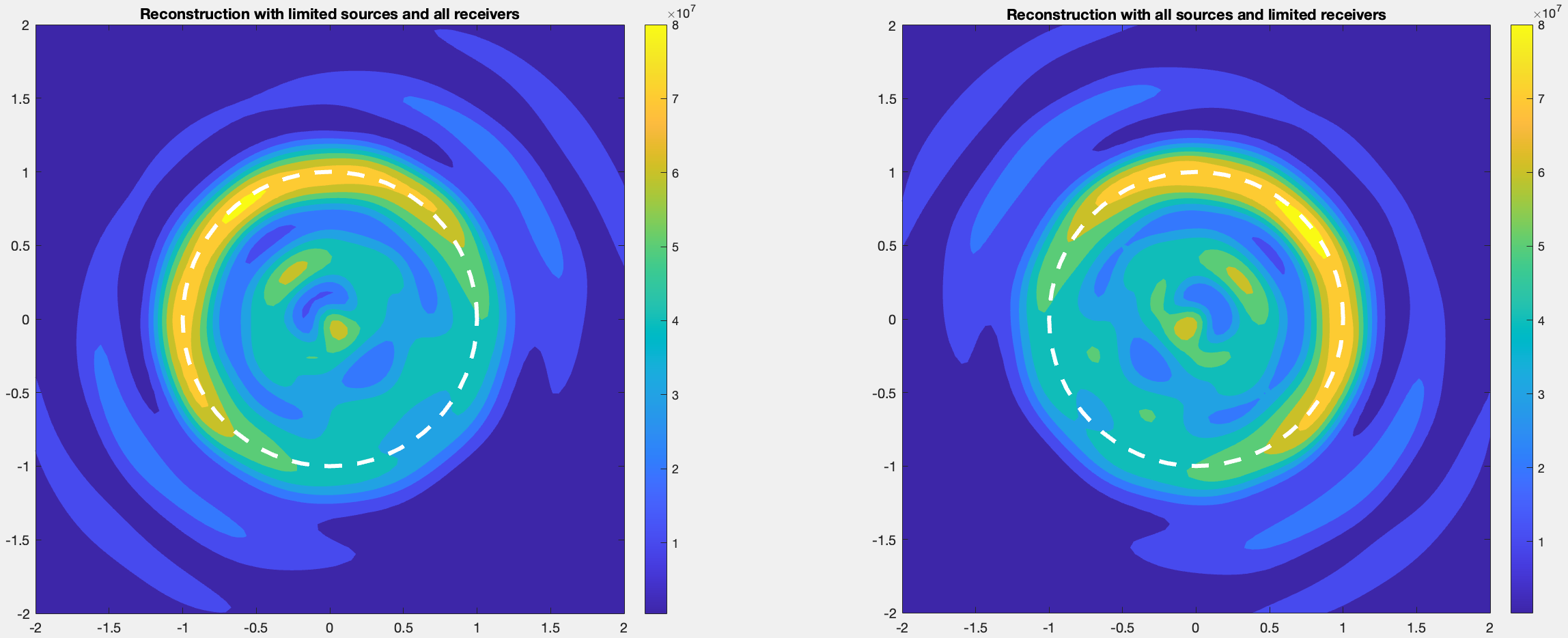}
\caption{Reconstruction of disk region using limited aperture with $10\%$ noise added by the DSM. Image left to right: reconstruction on limited aperture on the receivers and limited aperture on sources.}
\label{circlediffpalimited10}
\end{figure}

 In Figure \ref{circlediffpalimited10}, we see that having limited aperture information on the receivers or sources gives a good reconstruction even with random noise added to the data. The imaging functional that considers either limited receivers or sources still gives us a good reconstruction in terms of the location, shape, and size of the disk-shaped domain. \\
 
 \section{Conclusion}
In this paper, we investigated the inverse parameter and shape problem for an isotropic scatterer with two conductivity coefficients. This work considers an additional boundary parameter $\lambda$ and describes a novel method to recover the unknown scatterers. To achieve this, we developed a factorization of the far-field operator and then analyzed the operator to derive the new imaging function and showed that the imaging functional is stable with respect to noisy data. Furthermore, we addressed the uniqueness for recovering the coefficients from the known far-field data at a fixed incident direction for multiple frequencies. There are further questions to be explored for this scattering problem, such as: does the far-field data uniquely determine variable coefficients for the refractive index $n$ and boundary parameter $\eta$ when $\lambda$ is not fixed.\\

%
%

\noindent{\bf Acknowledgments:} The research of R. Ceja Ayala and I. Harris is partially supported by the NSF DMS Grant 2107891. 


\end{section}
\end{document}